\documentclass[11pt,a4paper]{article}

\usepackage{fullpage}
\usepackage[utf8x]{inputenc}
\usepackage{amsmath, amsthm}
\usepackage{wrapfig,graphicx,amssymb,textcomp,array,amsmath}
\usepackage{algpseudocode} 
\usepackage{enumerate}
\usepackage{enumitem}
\usepackage{multirow}
\usepackage{tabularx}
\algtext*{EndWhile}
\algtext*{EndIf}
\usepackage{algorithm}
\usepackage{color}

\newcommand{\area}[1]{A(#1)}
\newcommand{\dg}{G(n,r)}
\newcommand{\idg}[1]{DG[#1]}
\newcommand{\Ex}{\mathrm{E}}
\newcommand{\Var}{\mathrm{Var}}
\newcommand{\Cov}{\mathrm{Cov}}
\newcommand\given[1][]{\:#1\vert\:}

\date{\today}
\newtheorem{lemma}{Lemma}
\newtheorem{corollary}{Corollary}

\newtheorem{theorem}{Theorem}
\newtheorem{observation}{Observation}
\newtheorem*{problem*}{Problem}

\title{Plane and Planarity Thresholds for Random Geometric Graphs\thanks{A preliminary version of this paper appeared in ALGOSENSORS 2015.}}

\author{Ahmad Biniaz\thanks{Cheriton School of Computer Science, University of Waterloo, 
		Waterloo, Canada. This work has been done wile the author was a PhD student at Carleton University.}
\and Evangelos Kranakis\thanks{School of Computer Science, Carleton University, 
	Ottawa, Canada.}
\and Anil Maheshwari\footnotemark[2]
\and Michiel Smid\footnotemark[2]
}
\begin{document}
\maketitle

\begin{abstract}
A random geometric graph, $\dg$, is formed by choosing $n$ points independently and uniformly at random in a unit square; two
points are connected by a straight-line edge if they are at Euclidean distance at most $r$.
For a given constant $k$, we show that $n^{\frac{-k}{2k-2}}$ is a distance threshold function for $\dg$ to have a connected subgraph on $k$ points. Based on this, we show that $n^{-2/3}$ is a distance threshold for $\dg$ to be plane, and $n^{-5/8}$ is a distance threshold to be planar. We also investigate distance thresholds for $\dg$ to have a non-crossing edge, a clique of a given size, and an independent set of a given size.
\end{abstract}

\section{Introduction}
Wireless networks are usually modeled as disk graphs in the plane. Given a set $P$ of points in the plane and a positive parameter $r$, the {\em disk graph} is the geometric graph with vertex set $P$ which has a straight-line edge between two points $p,q\in P$ if and only if $|pq|\le r$, where $|pq|$ denotes the Euclidean distance between $p$ and $q$. If $r=1$, then the disk graph is referred to as {\em unit disk graph}.  
A {\em random geometric graph}, denoted by $\dg$, is a geometric graph formed by choosing $n$ points independently and uniformly at random in a unit square; two points are connected by a straight-line edge if and only if they are at Euclidean distance at most $r$, where $r=r(n)$ is a function of $n$ and $r \to 0$ as $n\to \infty$.

We say that two line segments in the plane {\em cross} each other if they have a point in common that is interior to both edges. Two line segments are {\em non-crossing} if they do not cross. Note that two non-crossing line segments may share an endpoint. A geometric graph is said to be {\em plane} if its edges do not cross, and {\em non-plane}, otherwise. An edge in a geometric graph is said to be {\em free} (or {\em non-crossing}) if its interior is not intersected by any other edge of the graph. A graph is {\em planar} if and only if it does not contain $K_5$ (the complete graph on 5 vertices) or $K_{3,3}$ (the complete bipartite graph on six vertices partitioned into two parts each of size $3$) as a minor. A {\em non-planar graph} is a graph which is not planar. A {\em clique} in a graph is a subset of vertices of the graph such that every two of them are adjacent. An {\em independent set} in a graph is a subset of vertices of the graph such that none of them are adjacent.

A graph property $\mathcal{P}$ is {\em increasing} if a graph $G$ satisfies $\mathcal{P}$, then by adding edges to $G$, the property $\mathcal{P}$ remains valid in $G$. Similarly, $\mathcal{P}$ is {\em decreasing} if a graph $G$ satisfies $\mathcal{P}$, then by removing edges from $G$, the property $\mathcal{P}$ remains valid in $G$. $\mathcal{P}$ is called a {\em monotone} property if $\mathcal{P}$ is either increasing or decreasing. Connectivity and ``having a clique of size $k$'' are increasing monotone properties, while planarity, ``being plane'', and ``having an independent set of size $k$'' are decreasing monotone properties. In this paper we will see that ``having a free edge'' is not a monotone property in $\dg$.

By~\cite{friedgut1996every, goel2004sharp} any monotone property of random geometric graphs has a {\em sharp threshold} function (see \cite{Bradonjic2014} for a definition). The thresholds in random geometric graphs are expressed by the distance $r$. In the sequel, the term w.h.p. (with high probability) is to be
interpreted to mean that the probability tends to
$1$ as $n \to \infty$. For an increasing property $\mathcal{P}$, the {\em threshold} is a function $t(n)$ such that if $r=o(t(n))$ then w.h.p. $\mathcal{P}$ does not hold in $\dg$, and if $r=\omega(t(n))$ then w.h.p. $\mathcal{P}$ holds in $\dg$. Symmetrically, for a decreasing property $\mathcal{P}$, the {\em threshold} is a function $t(n)$ such that if $r=o(t(n))$ then w.h.p. $\mathcal{P}$ holds in $\dg$, and if $r=\omega(t(n))$ then w.h.p. $\mathcal{P}$ does not hold in $\dg$. Note that a threshold function may not be unique. It is well known that $\sqrt{\ln n/n}$ is a connectivity threshold for $\dg$; see~\cite{Gupta1998,Panchapakesan2001,penrose1997longest}. In this paper we investigate (not necessarily sharp) thresholds in random geometric graphs for having a connected subgraph of constant size, being plane, and being planar.

\subsection{Related Work}

Random graphs were first defined and formally
studied by Gilbert in \cite{gilbert1959random}
and Erd\"{o}s and R\'{e}nyi \cite{erd6s1960evolution}.
It seems that the concept of a
random geometric graph was first formally suggested by Gilbert
in \cite{gilbert1961random}
and for that reason is also known as Gilbert's disk model. These
classes of graphs are known to have numerous applications
as a model for studying communication
primitives (broadcasting, routing, etc.) and
topology control (connectivity, coverage, etc.) in idealized wireless
sensor networks. They have been extensively studied in theoretical computer science and mathematical sciences over last few decades.

An instance of Erd\"{o}s-R\'{e}nyi graph~\cite{erd6s1960evolution} is obtained by taking $n$ vertices and connecting any two with probability $p$, independently of all other pairs; the graph derived by this scheme is denoted by $G_{n,p}$. In $G_{n,p}$ the threshold is expressed by the edge existence probability $p$, while in $\dg$ the threshold is expressed in terms of $r$. In both random graphs and random geometric graphs, property thresholds are of great interest~\cite{bollobas2001random,Bradonjic2014,friedgut1996every,goel2004sharp,Mccolm2001}.
Note that edge crossing configurations in $\dg$ have a geometric nature, and as such, have no analogues in the context of the Erd\"{o}s-R\'{e}nyi model for random graphs. However, planarity, and having a clique or an independent set of specific size are of interest in both $G_{n,p}$ and $\dg$.

Bollob\'{a}s and Thomason~\cite{Bollobas1987} showed that
any monotone property in random graphs has a threshold function. See also a result of Friedgut and Kalai~\cite{friedgut1996every}, and a result of Bourgain and Kalai~\cite{Bourgain1998}. In the Erd\"{o}s-R\'{e}nyi random graph
$G_{n,p}$, the connectivity threshold is $p = \log n/n$ and the threshold for having a giant component is $p= 1/n$; see~\cite{alon2004probabilistic}. The planarity threshold for $G_{n,p}$ is $p =1/n$, and its threshold for having a clique of size $k$ is $p=n^{-2/(k-1)}$; see~\cite{bollobas2001random,spencer1987ten}. 

A general reference on random geometric graphs is
\cite{penrose2003random}. There is extensive literature on various aspects of random geometric graphs
of which we mention the related work on coverage by
\cite{hall1985coverage,janson1986random}
and a review on percolation, connectivity, coverage and coloring by~\cite{balister2008percolation}.
As in random graphs, any monotone property in random geometric graphs has a threshold function~\cite{Bradonjic2014,goel2004sharp,Krishnamachari2002,Mccolm2001}. 

Random geometric graphs have a connectivity threshold of $\sqrt{\ln n/n}$; see~\cite{Gupta1998,Panchapakesan2001,penrose1997longest}.
Gupta and Kumar~\cite{Gupta1998} provided a connectivity threshold for points that are uniformly distributed in a disk. 
By a result of Penrose~\cite{Penrose1999}, in $\dg$, any threshold function for having no isolated vertex (a vertex of degree zero) is also a connectivity threshold function. Panchapakesan and Manjunath~\cite{Panchapakesan2001} showed that $\sqrt{\ln n/n}$ is a threshold for being an isolated vertex in $\dg$. This implies that $\sqrt{\ln n/n}$ is a connectivity threshold for $\dg$.
For $k\ge 2$, the details on the $k$-connectivity
threshold
in random geometric graphs can be found in~\cite{Penrose1999,penrose2003random}. Connectivity of random geometric graphs for points on a line is studied by Godehardt and Jaworski~\cite{Godehardt1996}.
Appel and Russo~\cite{Appel2002} considered the connectivity under the $L_\infty$-norm. 

\subsection{Our Results}
In this paper we investigate thresholds for some monotone properties in random geometric graphs. 
In Section~\ref{connected-section} we show that for a constant $k$, the distance threshold for having a connected subgraph on $k$ points is $n^{\frac{-k}{2k-2}}$. We show that the same threshold is valid for the existence of a clique of size $k$. Based on that, we prove the following thresholds for a random geometric graph to be plane or planar. In Section~\ref{plane-section}, we prove that $n^{-2/3}$ is a distance threshold for a random geometric graph to be plane. In Section~\ref{planar-section}, we prove that $n^{-5/8}$ is a distance threshold for a random geometric graph to be planar. In Section~\ref{free-edge-section} we investigate the existence of free edges in random geometric graphs. In Section~\ref{is-section}, we investigate thresholds for having an independent set of size $k$.

\section{The threshold for having a connected subgraph on $k$ points} 
\label{connected-section}
In this section, we look for the distance threshold for ``existence of connected subgraphs of constant size''; this is an increasing property. For a given constant $k$, we show that $n^{\frac{-k}{2k-2}}$ is the threshold function for the existence of a connected subgraph on $k$ points in $\dg$. Specifically, we show that if $r=o(n^{\frac{-k}{2k-2}})$, then w.h.p. $\dg$ has no connected subgraph on $k$ points, and if $r=\omega(n^{\frac{-k}{2k-2}})$, then w.h.p. $\dg$ has a connected subgraph on $k$ points. We also show that the same threshold function holds for the existence of a clique of size $k$.

\begin{theorem}
\label{connected-k-thr}
Let $k\ge 2$ be an integer constant. Then, $n^{\frac{-k}{2k-2}}$ is a distance threshold function for $\dg$ to have a connected subgraph on $k$ points.
\end{theorem}
\begin{proof}
Let $P_1, \dots, P_{n\choose k}$ be an enumeration of all subsets of $k$ points in $\dg$. Let $\idg{P_i}$ be the subgraph of $\dg$ that is induced by $P_i$. Let $X_i$ be the random variable such that

\[ X_i =
  \begin{cases}
    1       & \quad \text{if } \idg{P_i}\text{ is connected,}\\
    0     & \quad \text{otherwise.}\\
  \end{cases}
\]
Let the random variable $X$ count the number of sets $P_i$ for which $\idg{P_i}$ is connected. It is clear that

\begin{equation}
 \label{X-Xi}
X=\sum_{i=1}^{n\choose k}X_i.
\end{equation}
Observe that $\Ex[X_i]=\Pr[X_i=1]$. Since the random variables $X_i$ have identical distributions, we have
\begin{equation}
 \label{E-X-Xi}
\Ex[X]={n \choose k}\Ex[X_1].
\end{equation}

We obtain an upper bound and a lower bound for $\Pr[X_i=1]$. First, partition the unit square into squares of side equal to $r$. Let $\{s_1, \dots, s_{1/r^2}\}$ be the resulting set of squares. For a square $s_t$, let $S_t$ be the $kr\times kr$ square which has $s_t$ on its left bottom corner; see Figure~\ref{sS-fig}(a). $S_t$ contains at most $k^2$ squares each of side length $r$ ($S_t$ may be on the boundary of the unit square). Let $A_{i,t}$ be the event that all points in $P_i$ are contained in $S_t$. Observe that if $\idg{P_i}$ is connected then $P_i$ lies in $S_t$ for some $t\in\{1,\dots, 1/{r^2}\}$. Therefore, 

$$\mbox{if }\idg{P_i}\mbox{ is connected, then } (A_{i,1} \lor A_{i,2}\lor\dots\vee A_{i,1/{r^2}}),$$ and hence we have
\begin{equation}
 \label{Xi-upper}
\Pr[X_i=1]\le \sum_{t=1}^{1/{r^2}} \Pr[A_{i,t}]\le\sum_{t=1}^{1/{r^2}} (k^2r^2)^k=k^{2k} r^{2k-2}.
\end{equation}

Now, partition the unit square into squares with diagonal length equal to $r$. Each such square has side length equal to $r/\sqrt{2}$. Let $\{s_1, \dots, s_{2/r^2}\}$ be the resulting set of squares. Let $B_{i,t}$ be the event that all points of $P_i$ are in $s_t$. Observe that if all points of $P_i$ are in the same square, then $\idg{P_i}$ is a complete graph and hence connected. Therefore,
$$\mbox{if }(B_{i,1} \lor B_{i,2}\lor\dots\vee B_{i,2/{r^2}})\mbox{, then } \idg{P_i}\mbox{ is connected},$$ and hence we have
\begin{equation}
 \label{Xi-lower}
\Pr[X_i=1]\ge \sum_{t=1}^{2/{r^2}} \Pr[B_{i,t}]=\sum_{t=1}^{2/{r^2}} \left(\frac{r^2}{2}\right)^k=\frac{1}{2^{k-1}} r^{2k-2}.
\end{equation}

Since $k\ge 2$ is a constant, Inequalities~\eqref{Xi-upper} and~\eqref{Xi-lower} and Equation~\eqref{E-X-Xi} imply that 
\begin{align}
\Ex[X_i]&= \label{E-Xi} \Theta(r^{2k-2}),\\
\Ex[X] &=\label{E-X}\Theta(n^kr^{2k-2}).
\end{align}
If $n\to \infty$ and $r=o(n^{\frac{-k}{2k-2}})$ we conclude that the following inequalities are valid
\begin{align}
\Pr [X\ge 1] &\leq \notag \Ex[X]  \mbox{ (by Markov's Inequality)}\\
&= \notag \Theta(n^k r^{2k-2}) \mbox{ (by~\eqref{E-X})} \\
&= \label{1mom2:eq} o(1).
\end{align}
Therefore, w.h.p. $\dg$ has no connected subgraph on $k$ points.

\begin{figure}[htb]
 \centering
\setlength{\tabcolsep}{0in}
  $\begin{tabular}{cc}
\multicolumn{1}{m{.5\columnwidth}}{\centering\includegraphics[width=.24\columnwidth]{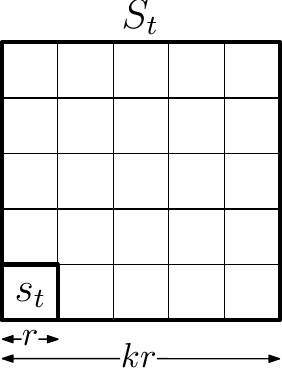}}
&\multicolumn{1}{m{.5\columnwidth}}{\centering\includegraphics[width=.24\columnwidth]{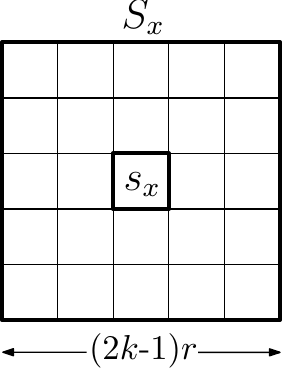}}
\\(a) & (b)
\end{tabular}$
  \caption{(a) The square $S_t$ has $s_t$ on its left bottom corner. (b) The square $S_x$ which is centered at $s_x$.}
\label{sS-fig}
\end{figure}

In the rest of the proof, we assume that $r=\omega(n^{\frac{-k}{2k-2}})$. In order to show that w.h.p. $\dg$ has at least one connected subgraph on $k$ vertices, we show, using the second moment method~\cite{alon2004probabilistic}, that $\Pr[X=0]\to 0$ as $n\to\infty$. 
Recall from Chebyshev's inequality that 
\begin{equation}
\label{Chebyshev}
\Pr [X=0] \leq 
\frac{\Var (X)}{\Ex[X]^2} .
\end{equation}
Therefore, in order to show that $\Pr [X=0] \to 0$, it suffices to show that
\begin{equation}
\label{Var-Ex2}
\frac{\Var (X)}{\Ex[X]^2}\to 0.
\end{equation}
In view of Identity~\eqref{X-Xi} we have
\begin{equation}
\label{Var-X}
\Var (X) = \sum_{1\le i,j\le {n\choose k}} \Cov (X_i, X_j),
\end{equation}
where $\Cov(X_i,X_j)=\Ex[X_iX_j]-\Ex[X_i]\Ex[X_j]\le \Ex[X_iX_j]$. 
If $|P_i\cap P_j|=0$ then $\idg{P_i}$ and $\idg{P_j}$ are disjoint. Thus, the random variables $X_i$ and $X_j$ are independent, and hence $\Cov(X_i,X_j)=0$. 
It is enough to consider the
cases when $P_i$ and $P_j$ are not disjoint. Assume $|P_i\cap P_j|=w$, where $w\in\{1,\dots,k\}$. Thus, in view of Equation~\eqref{Var-X}, we have 

\begin{align}
\Var (X) &=\notag \sum_{w=1}^k\sum_{|P_i\cap P_j|=w} \Cov (X_i, X_j)\\
& \le \label{Var-X-2} \sum_{w=1}^k\sum_{|P_i\cap P_j|=w} \Ex[X_iX_j].
\end{align}

The computation of $\Ex[X_i,X_j]$ involves some geometric considerations which are being discussed in detail below. Since $X_i$ and $X_j$ are 0-1 random variables, $X_iX_j$ is a 0-1 random variable and 

\[ X_iX_j =
  \begin{cases}
    1       & \quad \text{if both } \idg{P_i}\text{ and }\idg{P_j}\text{ are connected,}\\
    0     & \quad \text{otherwise.}\\
  \end{cases}
\]

By the definition of the expected value we have 
\begin{align}
\Ex[X_iX_j]&=\notag\Pr[X_j=1|X_i=1]\Pr[X_i=1]\\
&=\label{E-XiXj}\Pr[X_j=1|X_i=1]\Ex[X_i].
\end{align}

By~\eqref{E-Xi}, $\Ex[X_i]= \Theta(r^{2k-2})$. It remains to compute $\Pr[X_j=1|X_i=1]$, i.e., the probability that $\idg{P_j}$ is connected given that $\idg{P_i}$ is connected.
Consider the $k$-tuples $P_i$ and $P_j$ under the condition that $\idg{P_i}$ is connected. Let $x$ be a point in $P_i\cap P_j$.
Partition the unit square into squares of side length equal to $r$. Let $s_x$ be the square containing $x$. Let $S_x$ be the $(2k-1)r\times (2k-1)r$ square centered at $s_x$. $S_x$ contains at most $(2k-1)^2$ squares each of side length $r$ (if $S_x$ is on the boundary of the unit square then it may contain less than $(2k-1)^2$ squares); see Figure~\ref{sS-fig}(b). The area of $S_x$ is at most $(2kr)^2$, and hence the probability that a specific point of $P_j$ is in $S_t$ is at most $4k^2r^2$. 
Since $P_i$ and $P_j$ share $w$ points, in order for $\idg{P_j}$ to be connected, the remaining $k-w$ points of $P_j$ must lie in $S_x$. Thus, the probability that $\idg{P_j}$ is connected given that $\idg{P_i}$ is connected is at most $(4k^2r^2)^{k-w}\le c_wr^{2k-2w}$, for some constant $c_w>0$. Thus, $\Pr[X_j=1|X_i=1]\le c_wr^{2k-2w}$. In view of Equation~\eqref{E-XiXj}, we have

\begin{equation}
\label{E-XiXj-2}
 \Ex[X_iX_j]\le c'_w\cdot r^{2k-2w}\cdot r^{2k-2}= c'_wr^{4k-2w-2},
\end{equation}
for some constant $c'_w>0$.

Since $P_i$ and $P_j$ are $k$-tuples that share $w$ points, $|P_i\cup P_j|=2k-w$. There are ${n \choose {2k-w}}$ ways to choose $2k-w$ points for $P_i\cup P_j$. Since we choose $w$ points for $P_i\cap P_j$, $k-w$ points for $P_i$ alone, and $k-w$ points for $P_j$ alone, there are ${{2k-w}\choose {w, k-w, k-w}}$ ways to split the $2k-w$ chosen points into $P_i$ and $P_j$. Based on this and Inequality~\eqref{E-XiXj-2}, Inequality~\eqref{Var-X-2} turns out to
 
\begin{align}
\Var (X) & \le \notag \sum_{w=1}^k\sum_{|P_i\cap P_j|=w} \Ex[X_iX_j]\\
&\le \notag \sum_{w=1}^{k} {n \choose {2k-w}}{{2k-w}\choose {w, k-w, k-w}} c'_wr^{4k-2w-2}\\
&\le\notag \sum_{w=1}^{k} c''_w{n^{2k-w}}r^{4k-2w-2}.
\end{align}

for some constants $c''_w>0$. Consider~\eqref{Var-Ex2} and note that by~\eqref{E-X}, $\Ex[X]^2\ge c''n^{2k}r^{4k-4}$, for some constant $c''>0$. Thus,

\begin{align}
\frac{\Var (X)}{\Ex[X]^2} & \le \notag\sum_{w=1}^{k} \frac{c''_w{n^{2k-w}}r^{4k-2w-2}}{c''n^{2k}r^{4k-4}}
 = \notag\sum_{w=1}^{k} \frac{c''_w}{c''}\cdot\frac{1}{n^{w}r^{2w-2}}\\
& = \label{Var-Ex2-2}\frac{c''_1}{c''}\cdot\frac{1}{n^{1}r^{0}}+
\frac{c''_2}{c''}\cdot\frac{1}{n^{2}r^{2}}+
\dots+
\frac{c''_k}{c''}\cdot\frac{1}{n^{k}r^{2k-2}}
\end{align}

Since $r=\omega(n^{\frac{-k}{2k-2}})$, all terms in~\eqref{Var-Ex2-2} tend to zero. This proves the convergence in~\eqref{Var-Ex2}. Thus, $\Pr[X=0]\to 0$ as $n\to \infty$. This implies that if $r=\omega(n^{\frac{-k}{2k-2}})$, then $\dg$ has a connected subgraph on $k$ vertices with high probability.
\end{proof}

In the following theorem we show that if $k=O(1)$, then $n^{\frac{-k}{2k-2}}$ is also a threshold for $\dg$ to have a clique of size $k$; this is an increasing property.

\begin{theorem}
\label{clique-k-thr}
Let $k\ge 2$ be an integer constant. Then, $n^{\frac{-k}{2k-2}}$ is a distance threshold function for $\dg$ to have a clique of size $k$.
\end{theorem}
\begin{proof}
By Theorem~\ref{connected-k-thr}, if $r=o(n^{\frac{-k}{2k-2}})$, then w.h.p. $\dg$ has no connected subgraph on $k$ vertices, and hence it has no clique of size $k$. This proves the first statement. We prove the second statement by adjusting the proof of Theorem~\ref{connected-k-thr}, which is based on the second moment method. Assume
$r=\omega(n^{\frac{-k}{2k-2}})$. Let $P_1,\dots, P_{n\choose k}$ be an enumeration of all subsets of $k$ points. Let $X_i$ be equal to 1 if $\idg{P_i}$ is a clique, and 0 otherwise. Let $X=\sum X_i$. 

Partition the unit square into a set $\{s_1,\dots,s_{1/{r^2}}\}$ of squares of side length $r$. Let $S_t$ be the $2r\times 2r$ square which has $s_t$ on its left bottom corner. If $\idg{P_i}$ is a clique then $P_i$ lies in $S_t$ for some $t\in\{1,\dots, 1/{r^2}\}$. Therefore, 
$$
\Pr[X_i=1]\le 4^{k} r^{2k-2}.
$$
Now, partition the unit square into a set $\{s_1,\dots,s_{2/{r^2}}\}$ of squares with diagonal length $r$. If all points of $P_i$ fall in the square $s_t$, then $\idg{P_i}$ is a clique. Thus,
$$
\Pr[X_i=1]\ge \frac{1}{2^{k-1}} r^{2k-2}.
$$
Since $k\ge 2$ is a constant, we have
\begin{align}
\Ex[X_i]&=\notag \Theta(r^{2k-2}),\\
\Ex[X] &=\notag \Theta(n^kr^{2k-2}).
\end{align}

In view of Chebyshev's inequality we need to show that $\frac{\Var (X)}{\Ex[X]^2}$ tends to 0 as $n$ goes to infinity. We bound $\Var (X)$ from above by Inequality~\eqref{Var-X-2}.
Consider the $k$-tuples $P_i$ and $P_j$ under the condition that $\idg{P_i}$ is a clique. Let $|P_i\cap P_j|=w$, and let $x$ be a point in $P_i\cap P_j$.
Partition the unit square into squares of side length $r$. Let $s_x$ be the square containing $x$. Let $S_x$ be the $3r\times 3r$ square centered at $s_x$. In order for $\idg{P_j}$ to be a clique, the remaining $k-w$ points of $P_j$ must lie in $S_x$. Thus, 
$$
 \Ex[X_iX_j]\le c'_wr^{4k-2w-2},
$$
for some constant $c'_w>0$. By a similar argument as in the proof of Theorem~\ref{connected-k-thr}, we can show that for some constants $c'', c''_w>0$ the followings inequalities are valid: 
\begin{align}
\Var (X) &\le\notag \sum_{w=1}^{k} c''_w{n^{2k-w}}r^{4k-2w-2},\\
\frac{\Var (X)}{\Ex[X]^2} &\le\notag \sum_{w=1}^{k} \frac{c''_w}{c''}\cdot\frac{1}{n^{w}r^{2w-2}}.
\end{align}
Since $r=\omega(n^{\frac{-k}{2k-2}})$, the last inequality tends to 0 as $n$ goes to infinity. This completes the proof for the second statement.
\end{proof}
As a direct consequence of Theorem~\ref{clique-k-thr}, we have the following corollary.

\begin{corollary}
\label{clique-cor}
 $n^{-1}$ is a threshold for $\dg$ to have an edge, and $n^{-\frac{3}{4}}$ is a threshold for $\dg$ to have a triangle.
\end{corollary}

\section{The threshold for $\dg$ to be plane}
\label{plane-section}

In this section we investigate the 
threshold for a random geometric graph to be plane; this is a decreasing property. Recall that $\dg$ is plane if no two of its edges cross.
As a warm-up exercise we first prove a simple result which is based on the connectivity threshold for random geometric graphs, which is known to be $\sqrt{\ln n/n}$.\\
\begin{wrapfigure}{r}{0.3\textwidth}
\vspace{-15pt}
  \centering
\includegraphics[width=.25\textwidth]{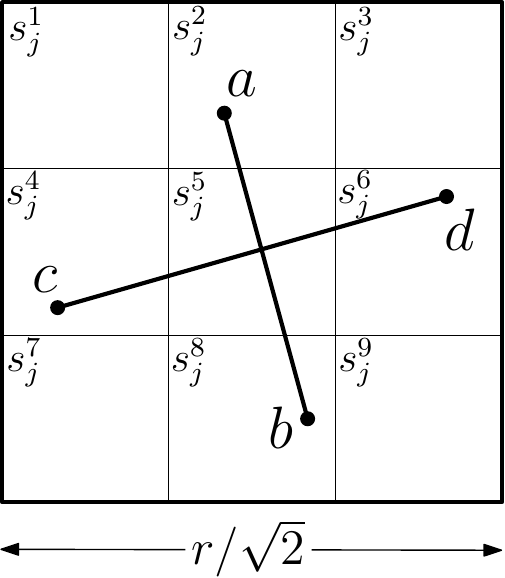}
  \vspace{-5pt}
\caption{An square of diameter $r$ which is partitioned into nine sub-squares.}
\vspace{-15pt}
\label{fig:square}
\end{wrapfigure}
 \vspace{-15pt}
\begin{theorem}
\label{thm3}
If $r \geq \sqrt{\frac{c \ln n}{n}}$, with $c \geq 36$, then w.h.p. $\dg$ is not plane.
\end{theorem}

\begin{proof}
In order to prove that w.h.p. $\dg$ is not plane, we show that w.h.p. it has a pair of crossing edges. Partition the unit square into squares each with diagonal length $r$. Then subdivide each such square into nine sub-squares as depicted in Figure~\ref{fig:square}.
There are $\frac{18}{r^2}$ sub-squares, each of side length $\frac{r}{3 \sqrt{2}}$. The probability that no point lies in a specific sub-square is $(1 - \frac{r^2}{18})^n$.
Thus, the probability that there exists an empty sub-square
is at most
$$
\frac{18}{r^2} \left(1 - \frac{r^2}{18} \right)^n
\leq n \left(1 - \frac{c \ln n}{18 n} \right)^n \leq n^{1-c/18}
\leq \frac{1}{n},
$$
when $c \geq 36$.
Therefore, with probability at least $1-\frac{1}{n}$ all sub-squares contain points. By choosing four points $a$, $b$, $c$, and $d$ as depicted in Figure~\ref{fig:square}, it is easy to see that the edges $(a,b)$ and $(c,d)$ cross. 
Thus, w.h.p. $\dg$ has a pair of crossing edges, and hence w.h.p. it is not plane.
\end{proof}

In fact, Theorem~\ref{thm3} ensures that w.h.p. there exists a pair of crossing edges in each of the squares. This implies that there are $\Omega \left( \frac{n}{\ln n} \right)$ disjoint pair of crossing edges, while for $\dg$ to be not plane we need to show the existence of at least one pair of crossing edges. 
Thus, the value of $r$ provided by the connectivity
threshold seems rather weak. 
By a different approach, in the rest of this section we show that $n^{-\frac{2}{3}}$ is the correct threshold. 

\begin{lemma}
\label{clm1}
Let $(a,b)$ and $(c,d)$ be two crossing edges in $\dg$, and let $Q$ be the convex quadrilateral formed by $a$, $b$, $c$, and $d$. Then, two adjacent sides of $Q$ are edges of $\dg$.
\end{lemma}
\begin{proof}
Refer to Figure~\ref{fig:square1}.
At least one of the angles of $Q$, say $\angle cad$, is bigger than or equal to $\pi/2$.
It follows that in the triangle $\triangle cad$ the side $cd$ is the longest, i.e., $|cd|\ge \max\{|ac|,|ad|\}$. Since $|cd|\le r$, both $|ac|$ and $|ad|$ are at most $r$. Thus, $ac$ and $ad$\textemdash which are adjacent\textemdash are edges of $\dg$.
\end{proof}

\begin{figure}[htb]
  \centering
\setlength{\tabcolsep}{0in}
  $\begin{tabular}{cc}
\multicolumn{1}{m{.5\columnwidth}}{\centering\includegraphics[width=.3\columnwidth]{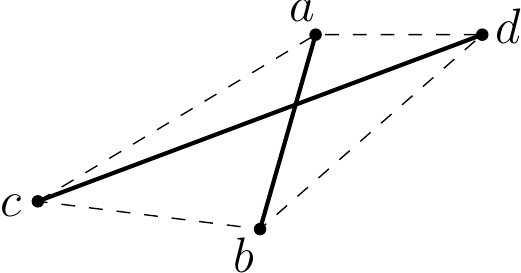}}
&\multicolumn{1}{m{.5\columnwidth}}{\centering\includegraphics[width=.3\columnwidth]{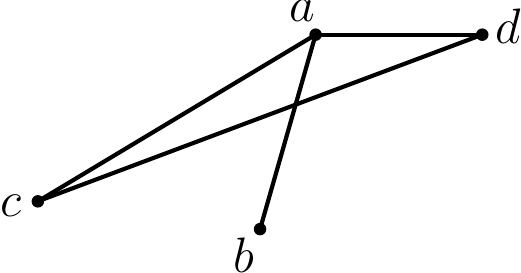}}
\\
(a) & (b)
\end{tabular}$
  \caption{(a) Illustration of Lemma~\ref{clm1}. (b) Crossing edges $(a,b)$ and $(c,d)$ form an anchor.}
\label{fig:square1}
 \vspace{-10pt}
\end{figure}

In the proof of Lemma~\ref{clm1}, $a$ is connected to $b$, $c$, and $d$. So the distance between $a$ to each of $b$, $c$, and $d$ is at most $r$. Thus, we have the following corollary.

\begin{corollary}
\label{2-hop-cor}
The endpoints of every two crossing edges in $\dg$ are at distance at most $2r$ from each other. Moreover, there exists an endpoint which is within distance $r$ from other endpoints.
\end{corollary}

Based on the proof of Lemma~\ref{clm1}, we define an {\em anchor} as a set $\{a, b, c, d\}$ of four points in $\dg$ such that three of them form a triangle, say $\triangle cad$, and the fourth vertex, $b$, is connected to $a$ by an edge which crosses
$cd$; see Figure~\ref{fig:square1}(b). We call $a$ as the {\em crown} of the anchor. The crown is within distance $r$ from the other three points. Note that $bc$ and $bd$ may or may not be edges of $\dg$.
In view of Lemma~\ref{clm1}, two crossing
edges in $\dg$ form an anchor. Conversely, every anchor in $\dg$ introduces a pair of crossing edges. 

\begin{observation}
\label{anchor-obs}
$\dg$ is plane if and only if it has no anchor.
\end{observation}

\begin{theorem}
\label{plane-thr}
$n^{-\frac{2}{3}}$ is a threshold for $\dg$ to be plane.
\end{theorem}
\begin{proof}
In order to show that $\dg$ is plane, by Observation~\ref{anchor-obs}, it is enough to show that it has no anchors. Every anchor has four points and it is connected. By Theorem~\ref{connected-k-thr}, if $r=o(n^{-\frac{2}{3}})$, then w.h.p. $\dg$ has no connected subgraph on $4$ points, and hence it has no anchors.
This proves the first statement. 

We prove the second statement by adjusting the proof of Theorem~\ref{connected-k-thr} for $k=4$. Assume
$r=\omega(n^{-\frac{2}{3}})$. Let $P_1,\dots, P_{n\choose 4}$ be an enumeration of all subsets of $4$ points. Let $X_i$ be equal to 1 if $\idg{P_i}$ contains an anchor, and 0 otherwise. Let $X=\sum X_i$. In view of Chebyshev's inequality we need to show that $\frac{\Var (X)}{\Ex[X]^2}$ tends to 0 as $n$ goes to infinity.

Partition the unit square into a set $\{s_1,\dots,s_{2/{r^2}}\}$ of squares with diagonal length $r$. Then, subdivide each square $s_j$, into nine sub-squares $s_j^1,\dots, s_j^9$ as depicted in Figure~\ref{fig:square}. If each of $s_j^1,s_j^3,s_j^7, s_j^9$ or each of $s_j^2,s_j^4,s_j^6, s_j^8$ contains a point of $P_i$, then $\idg{P_i}$ is a convex clique of size four and hence it contains an anchor. Thus,
$$
\Pr[X_i=1]\ge \frac{r^{6}}{2^{3}}\cdot \frac{2}{9^4}.
$$
This implies that $\Ex[X_i]=\Omega(r^{6})$, and hence $\Ex[X] =\Omega(n^4r^6)$. Therefore, $$\Ex[X]^2\ge c''n^8r^{12},$$ for some constant $c''>0$.
By a similar argument as in the proof of Theorem~\ref{connected-k-thr} we bound the variance of $X$ from above by $$\Var(X)\le c''_1 n^{7}r^{12}+c''_2 n^{6}r^{10}+c''_3 n^{5}r^{8}+c''_4 n^{4}r^{6}.$$

Since $r=\omega(n^{-\frac{2}{3}})$, $\frac{\Var (X)}{\Ex[X]^2}$ tends to 0 as $n$ goes to infinity. That is, w.h.p. $\dg$ has an anchor.
By Observation~\ref{anchor-obs}, w.h.p. $\dg$ is not plane. 
\end{proof}

As a direct consequence of the proof of Theorem~\ref{plane-thr}, we have the following:
\begin{corollary}
With high probability if a random geometric graph is not plane, then it has a clique of size four.
\end{corollary}

Note that every anchor introduces a crossing and each crossing introduces an anchor. Since, every anchor is a connected graph and has four points, by~\eqref{E-X} we have the following corollary.
\begin{corollary}
\label{cr:cor}
The expected number of crossings in $\dg$ is $\Theta(n^4 r^6)$.
\end{corollary}

\section{The threshold for $\dg$ to be planar}
\label{planar-section}

In this section we investigate the threshold for the planarity of a random geometric graph; this is a decreasing property. By Kuratowski's theorem, a finite graph is planar if and only if it does not contain a subgraph that is a subdivision of $K_5$ or of $K_{3,3}$. Note that any plane random geometric graph is planar too; observe that the reverse statement may not be true. Thus, the threshold for planarity seems to be larger than the threshold of being plane. By a similar argument as in the proof of Theorem~\ref{thm3} we can show that if $r\ge \sqrt{{c \ln n}/{n}}$, then w.h.p. each square with diagonal length $r$ contains $K_5$, and hence $\dg$ is not planar.

\begin{theorem}
\label{planarity-thr}
$n^{-\frac{5}{8}}$ is a threshold for $\dg$ to be planar.
\end{theorem}
\begin{proof}
By Theorem~\ref{clique-k-thr}, if $r=\omega(n^{-\frac{5}{8}})$, then w.h.p. $\dg$ has a clique of size $5$. Thus, w.h.p. $\dg$ contains $K_5$ and hence it is not planar. 

If $r=o(n^{-\frac{5}{8}})$, then by Theorem~\ref{connected-k-thr}, w.h.p. $\dg$ has no connected subgraph on $5$ points, and hence it has no $K_5$. Similarly, if $r=o(n^{-\frac{3}{5}})$, then  w.h.p. $\dg$ has no connected subgraph on $6$ points, and hence it has no $K_{3,3}$. 
Since $n^{-\frac{5}{8}} < n^{-\frac{3}{5}}$, it follows that if $r=o(n^{-\frac{5}{8}})$, then w.h.p. $\dg$ has neither $K_5$ nor $K_{3,3}$ as a subgraph. 

Note that, in order to prove that $\dg$ is planar, we have to show that it does not contain any subdivision of either
$K_5$ or $K_{3,3}$. Any subdivision of either $K_5$ or $K_{3,3}$ contains a connected subgraph on $k\ge 5$ vertices. Since $n^{-5/8} < n^{-k/(2k-2)}$ for all $k \geq 5$, in view of Theorem~\ref{connected-k-thr}, we conclude that if $r =o(n^{-\frac{5}{8}})$, then w.h.p. $\dg$ has no subdivision of $K_5$ and $K_{3,3}$, and hence $\dg$ is planar. 
\end{proof}

As a direct consequence of the proof of Theorem~\ref{planarity-thr}, we have the following:
\begin{corollary}
With high probability if a random geometric graph does not contain a clique of size five, then it is planar.
\end{corollary}

\section{Free edges in $\dg$}
\label{free-edge-section}
Motivated by the problem of finding bottleneck plane matching of a point set (see \cite{Abu-Affash2015}), we considered the problem of finding a free edge in unit disk graphs. A {\em free edge} in $\dg$ is an edge whose interior is not intersected by other edges of $\dg$. 
Although at the first glance it seems that any unit disk graph has at least one free edge, this is not always true. Figure~\ref{RB-fig} shows examples of unit disk graphs that do not have any free edges, i.e., all the edges are crossed. These examples imply that the ``existence of a free edge'' is not a monotone property for $\dg$. In this section we investigate the existence of free edges in random geometric graphs. Throughout this section we assume that points are chosen on a torus, i.e., a unit square with wraparound.  

\begin{figure}[htb]
  \centering
\setlength{\tabcolsep}{0in}
  $\begin{tabular}{ccc}
\multicolumn{1}{m{.35\columnwidth}}{\centering\includegraphics[width=.34\columnwidth]{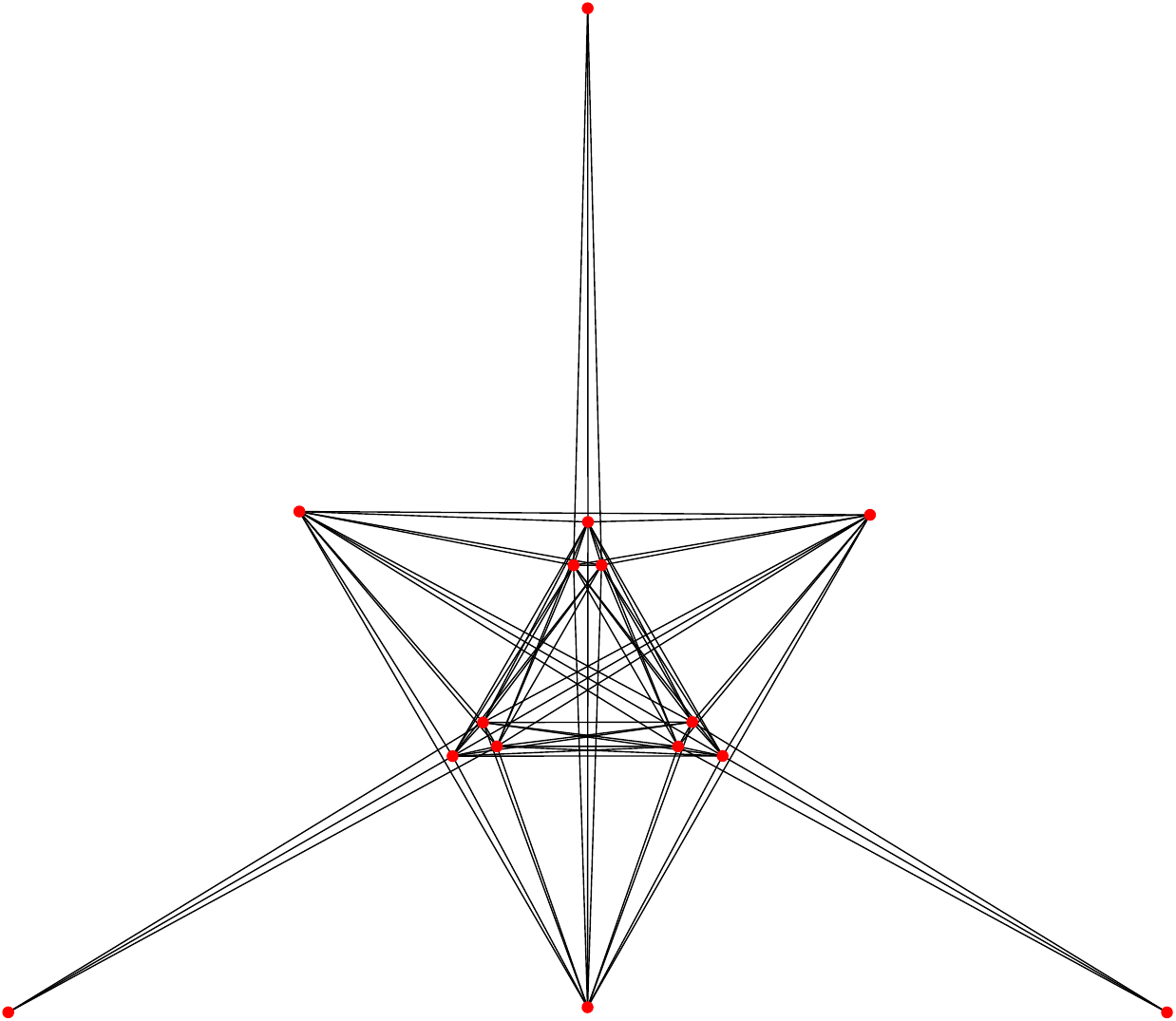}}
&\multicolumn{1}{m{.28\columnwidth}}{\centering\includegraphics[width=.27\columnwidth]{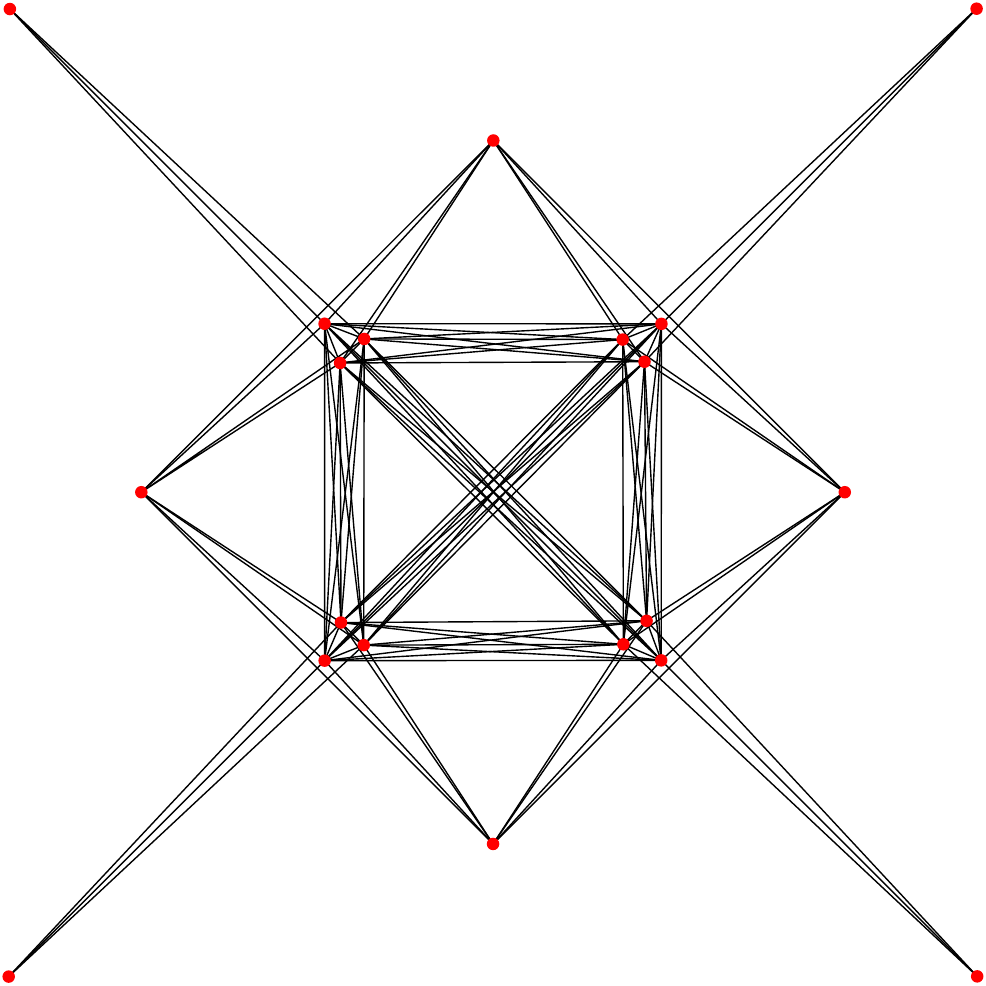}}
&\multicolumn{1}{m{.37\columnwidth}}{\centering\includegraphics[width=.36\columnwidth]{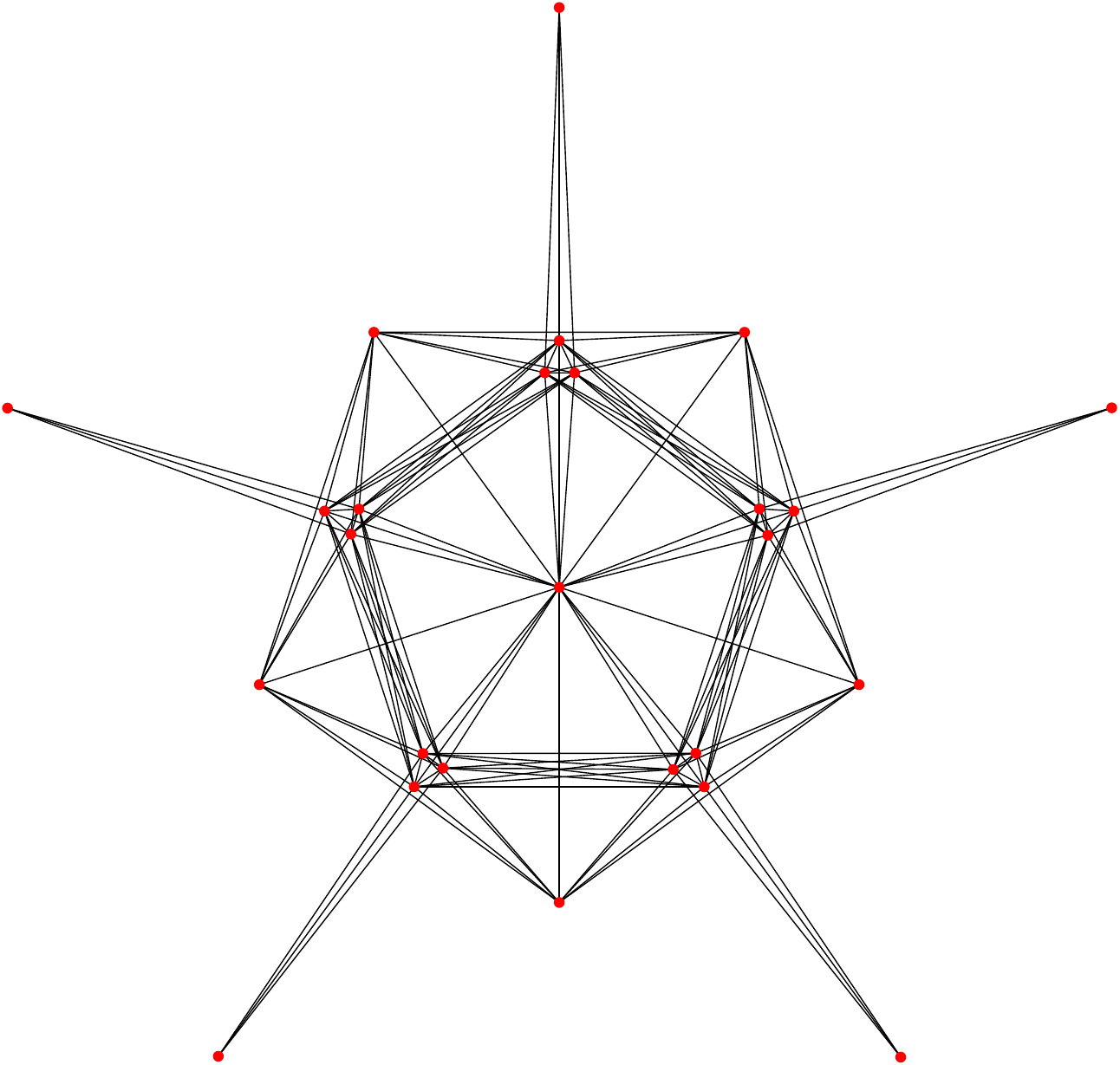}}
\end{tabular}$
\vspace{-5pt}
  \caption{Unit disk graphs with 15, 20, and 26 points, that do not have any free edge.}
\label{RB-fig}
\end{figure}

\begin{lemma}
\label{free-edge-lemma1}
If $r=\omega\left(\frac{1}{n}\right)$ and $r=o\left(\frac{1}{n^{1/2+\epsilon}}\right)$, for any constant $\epsilon>0$, then $\dg$ has a free edge with high probability.
\end{lemma}
\begin{proof}
Since $r=\omega\left(\frac{1}{n}\right)$, by Corollary~\ref{clique-cor}, $\dg$ has at least one edge with high probability. 

Let $k\ge 3$ be an integer constant such that $\frac{k}{2k-2}\le\frac{1}{2}+\epsilon$.
Then $r=o(n^{\frac{-k}{2k-1}})$. By Theorem~\ref{connected-k-thr}, $\dg$ has no connected component of size $k$ with high probability. Moreover, there exists an integer constant $c$, where $2\le c < k$, such that $r=\omega(n^{\frac{-c}{2c-2}})$ and $r=o(n^{\frac{-(c+1)}{2(c+1)-2}})$. Theorem~\ref{clique-k-thr} implies that with high probability $\dg$ has a clique of size $c$. Let $H$ be clique. We claim that all the edges of
the boundary of the convex hull of $H$ are free. For contradiction assume otherwise, and take an edge $(a,b)$ from the boundary of the convex hull
of $H$ that is crossed by some other edge $(c,d)$. Note that at least one of $c$ and $d$ is not a vertex of $H$. Let $d$ be a vertex that is not in $H$. By Corollary~\ref{2-hop-cor},  in $\dg$, there is a path (consists of at most two edges) between $d$ and each of $a$ and $b$. This in turn gives a connected subgraph with $c+1$ vertices in $\dg$. However, by Theorem~\ref{connected-k-thr}, such a connected component cannot exists with high probability. This completes the proof of the theorem.
\end{proof}

We refer to an edge $(u,v)$ in $\dg$ as a {\em long edge} if $|uv| \geq \sqrt{\frac{8\ln n}{n}}$, where $|uv|$ denotes the Euclidean distance between $u$ and $v$. For a given region $R$ in the unit square let $\area{R}$ denote the area of $R$. Let $P$ denote the set of $n$ random points.

\begin{lemma}
\label{free-edge-lemma2}
With high probability, none of the long edges of $\dg$ are free.
\end{lemma}
\begin{proof} 
Let $L$ be the set consisting of all the long edges of $\dg$. Note that $L$ contains at most $n\choose 2$ edges. Take any edge $(u,v)$ in $L$. Consider the closed disk that has $(u,v)$ as a diameter. Let $H_1(u,v)$ and $H_2(u,v)$ be the two half-disks on both sides of $(u,v)$. We say that $H_1(u,v)$ (resp. $H_2(u,v)$) is {\em empty} if it does not contain any point of $P\setminus\{u,v\}$. Observe that if both $H_1(u,v)$ and $H_2(u,v)$ are non-empty, then $(u,v)$ is not a free edge.
We are going to show that with high probability both $H_1(u,v)$ and $H_2(u,v)$ are non-empty.
 
Let $\mathcal{H}$ be the set of all half-disks defined by the edges in $L$. Let $H$ be any half-disk in $\mathcal{H}$, and let $d$ be the diameter of $H$. Then, $\area{H}=\frac{\pi d^2}{8}$. Thus, with probability $\left(1 - \frac{\pi d^2}{8}\right)^{n-2}$, $H$ is empty.
Since $d\ge \sqrt{8\ln n/n}$, we have 
$$\left(1 - \frac{\pi d^2}{8}\right)^{n-2} \le \left(1 - \frac{\pi \ln n}{n}\right)^{n-2} 
\le e^{\frac{-\pi(n-2)\ln n}{n}}\le e^{-3\ln n}=\frac{1}{n^3},$$
 where the last inequality is valid because $\pi(n-2)>3n$ as $n\to \infty$. Therefore, with probability at most $\frac{1}{n^3}$, $H$ is empty. Since $\mathcal{H}$ contains at most $n^2$ half-disks, the probability that $\mathcal{H}$ contains an empty half-disk is at most $n^2\cdot\frac{1}{n^3}=\frac{1}{n}$. Thus, with probability at least $1-\frac{1}{n}$ all half-disks in $\mathcal{H}$ are non-empty. Therefore, both $H_1(u,v)$ and $H_2(u,v)$ contain points from $P\setminus\{u,v\}$. This produces an edge in $\dg$ that crosses $(u,v)$. Thus, with high probability all the long edges in $\dg$ are crossed.
\end{proof}

In the following lemma we prove a result that is stronger than the result of Lemma~\ref{free-edge-lemma2}.
\begin{lemma}
\label{free-edge-lemma3}
If
$r = \Omega \left( \frac{1}{\sqrt{s}}\left(\frac{\ln n}{n}\right)^{3/4} \right)$,
then $s=\Omega\left(\left(\frac{\ln n}{n}\right)^{3/2}\right)$, and with high probability, every edge of $\dg$ of length at least $s$ is not free.
\end{lemma}
\begin{proof} 
Since $r\le 1$, we have $\sqrt{s}>\left(\frac{c\ln n}{n}\right)^{3/4}$ for some constant $c>0$. This implies the first part of the claim, that is $s=\Omega\left(\left(\frac{\ln n}{n}\right)^{3/2}\right)$.

Now we show that every edge of length at least $s$ is crossed. Take any edge $(u, v)$ of length at least $s$ from $\dg$. If $|uv|\geq\sqrt{8\ln n/n}$, then $(u,v)$ is a long edge, and by Lemma~\ref{free-edge-lemma2}, $(u,v)$ is not free with high probability. 

Assume $|uv|<\sqrt{8\ln n/n}$, and thus, $s<\sqrt{8\ln n/n}$. This implies that $r\ge c'\sqrt{\ln n/n}$ for some constant $c'>0$. In order to show that $(u,v)$ is crossed, we are going to define two regions, $V_1$ and $U$, on opposite sides of $(u,v)$ such that the area of each region is least $\frac{\pi\ln n}{n}$. Then we show that each of $V_1$ and $U$ contains a point of $P$ such that these two points are connected by an edge that crosses $(u,v)$.

Without loss of generality assume that $(u,v)$ is vertical, and let $x$ be its midpoint. Let $V_1$ be a sector of angle $\frac{\pi}{3}$ to the left of oriented edge $(u,v)$ that has radius $\sqrt{6 \ln n/n}$ and is centered at $x$ such that its sides make angle $\frac{\pi}{3}$ with $(u,v)$; see Figure~\ref{fig:radius}. Note that $\area{V_1}=\frac{\pi\ln n}{n}$. As in the proof of Lemma~\ref{free-edge-lemma2}, with high probability $V_1$ contains a point of $P$. Let $w$ be such a point.
\begin{figure}[htb]
\begin{center}
\includegraphics[width=0.5\textwidth]{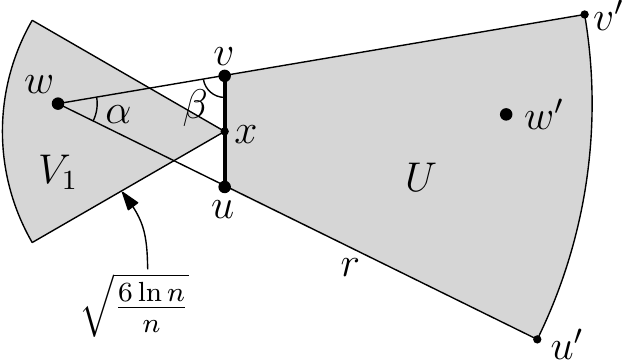}
\end{center}
\caption{Illustration of Lemma~\ref{free-edge-lemma3}. Each of the shaded regions has area at least $\frac{\pi \ln n}{n}$.}
\label{fig:radius}
\end{figure}

Extend the line segments $wu$ and $wv$ to $u'$ and $v'$,
respectively, such that $|wu'|=|wv'|=r$. Let $V_2$ be the sector of radius $r$ that is centered at $w$ and has $wu'$ and $wv'$ as its sides. Let $U$ be the subset of $V_2$ that is to the right of oriented edge $(u,v)$. Note that if $U$ contains a point $w'\in P$, then $(w,w')$ is an edge of $\dg$ and crosses $(u,v)$; see Figure~\ref{fig:radius}.

We would like the radius $r$ to be sufficiently large such that the area of $U$ be at least $\frac{\pi\ln n}{n}$. Without loss of generality assume that $w$ is closer to $v$ than to $u$, i.e., $|wu|\ge |wv|$. 
Since $|uv|<\sqrt{8\ln n/n}$, we have $|ux|=|vx|< \sqrt{2\ln n/n}$. Moreover $|wu|< |wx|+|ux|< \sqrt{6\ln n/n}+\sqrt{2\ln n/n}<4\sqrt{\ln n/n}$. By picking $c'\ge 8$ we have $r\ge 8\sqrt{\ln n/n}$; this makes sure that $r> 2 |wu|$, and subsequently $r> 2|wv|$. Thus, $\area{U}> \frac{1}{2}\area{V_2}$. Let $\alpha=\angle uwv$. Then, $\area{V_2}=\frac{\alpha r^2}{2}$ and $\area{U}>\frac{\alpha r^2}{4}$. 
If $\alpha\geq \frac{\pi}{3}$, then $$\area{U}> \frac{\pi r^2}{12}\ge\frac{64 \pi \ln n}{12n}>\frac{\pi\ln n}{n}.$$

Assume $\alpha< \frac{\pi}{3}$. Since $s>0$, we have $\alpha>0$. Let $\beta=\angle wvu$. Observe that $\frac{\pi}{3}<\beta<\frac{2\pi}{3}$. By the law of sines we have $$\sin \alpha=\sin\beta\cdot\frac{|uv|}{|wu|} >\frac{\sqrt{3}}{2}\cdot\frac{|uv|}{|wu|}\ge\frac{\sqrt{3}}{2}\cdot\frac{s}{|wu|}.$$ 
Note that $\sin\alpha<\alpha$, for all $0<\alpha< \frac{\pi}{3}$. Thus, $$\alpha >\frac{\sqrt{3}}{2}\cdot\frac{s}{|wu|},$$
and hence $$\area{U}>\frac{\alpha r^2}{4}>\frac{\sqrt{3}}{8}\cdot\frac{sr^2}{|wu|} >\frac{\sqrt{3}}{32}\cdot\frac{sr^2}{\sqrt{\ln n/n}}.$$

To ensure $\area{U}\ge\frac{\pi\ln n}{n}$,
it is enough to guarantee the following inequality
\begin{equation}
 \label{small-s}
\frac{\sqrt{3}}{32}\cdot\frac{sr^2}{\sqrt{\ln n/n}}  \geq 
\frac{\pi\ln n}{n}.
\end{equation}

However, Inequality~\eqref{small-s} is valid if 
$$
r \geq \frac{c}{\sqrt{s}} \cdot \left( \frac{\ln n}{n} \right)^{\frac{3}{4}},
$$
for some constant $c \ge 8$. This is ensured by the statement of the lemma.

To this end, for each edge $(u,v)$ in $\dg$ of length at least $s$, we have two regions $V_1$ and $U$ such that each of them has area at least $\frac{\pi\ln n}{n}$.
Let $\mathcal{R}$ be the set of all such regions for all edges of $\dg$ whose lengths are at least $s$. Note that $\mathcal{R}$ contains at most $n^2$ regions. As in the proof of Lemma~\ref{free-edge-lemma2}, the probability that a region $R\in \mathcal{R}$ does not contain any point of $P$ is at most $\frac{1}{n^3}$, and thus, the probability that there exists an empty region in $\mathcal{R}$ is at most $n^2\cdot\frac{1}{n^3}=\frac{1}{n}$. Thus, with probability at least $1-\frac{1}{n}$ all regions in $\mathcal{R}$ have a point of $P$. Therefore, both $V_1$ and $U$ contain points from $P$. This produces an edge in $\dg$ that crosses $(u,v)$. Thus, with high probability all the edges of $\dg$ of length at least $s$ are crossed.
\end{proof}

\begin{theorem}
\label{free-edge-thr}
Let $\dg$ be a random geometric graph.
\begin{enumerate}
\item
If $r \ll\frac{1}{n}$ then w.h.p. $\dg$ has no edge.
\item
If $\frac1n \ll r \ll \frac1{n^{2/3}}$ then w.h.p. every edge of $\dg$ is free.
\item
If $\frac{1}{n^{2/3}} \ll r \ll \frac{1}{n^{1/2+\epsilon}}$, for any constant $\epsilon>0$, w.h.p. $\dg$ has a free edge.
\item
If $r \gg \frac{(\ln n)^{3/4}}{n^{1/4}}$, then w.h.p. $\dg$ has no free edge.
\end{enumerate}
\end{theorem}
\begin{proof}
The first part is immediate from Corollary~\ref{clique-cor}.
For the second part, in this range there exists an edge by Corollary~\ref{clique-cor}, and there is no crossing by Theorem~\ref{plane-thr}.  
The third part has been proved in Lemma~\ref{free-edge-lemma1}.
It remains to prove the last part.
By Corollary~\ref{clique-cor}, with high probability no edge of the $\dg$ can have length $o(\frac{1}{n})$. Thus, with high probability all edges of $\dg$ have length $\Omega(\frac{1}{n})$. Therefore, the result
follows from Lemma~\ref{free-edge-lemma3}
by setting $s = \frac{1}{n}$. Note that for large values of $n$ we have $\frac{1}{n}= \Omega\left(\left(\frac{\ln n}{n}\right)^{3/2}\right)$, which is required in Lemma~\ref{free-edge-lemma3}.
\end{proof}

\section{Thresholds for having an independent set of size $k$}
\label{is-section}
In this section we investigate thresholds for a random geometric graph to have an independent set of a given size; this
is a decreasing property. Throughout this section
we assume that points are chosen on a torus, i.e., a unit square with wraparound.
Note that if $\dg$ has no edge then it has an independent set of size $n$. Moreover, if $\dg$ has at least one edge, then it has no independent set of size $n$. By Corollary~\ref{clique-cor}, $n^{-1}$ is a radius threshold for $\dg$ to have an edge. Therefore, $n^{-1}$ is a radius threshold for $\dg$ to have an independent set of size $n$.
\begin{theorem}
\label{is-thr1}
Let $k\geq \ln n$ be an integer. If $r>\sqrt{\frac{c\ln(\frac{en}{k})}{k}}$, for some constant $c>0$, then $\dg$ has no independent set of size $k$ with high probability.
\end{theorem}
\begin{proof} 
Let $P_1, \dots, P_{n\choose k}$ be an enumeration of all subsets of $k$ points in $\dg$. Let $X_i$ be a random variable such that

\[ X_i =
  \begin{cases}
    1       & \quad \text{if } P_i\text{ is an independent set,}\\
    0     & \quad \text{otherwise.}\\
  \end{cases}
\]
Let the random variable $X$ counts the number of independent sets $P_i$. Thus, $X=\sum_{i=1}^{n\choose k}X_i.$
Since $X_i$'s have identical distributions, we have
\begin{equation}
 \label{is-E-X-X1}
\Ex[X]={n \choose k}\Ex[X_1].
\end{equation}

Observe that $\Ex[X_1]=\Pr[X_1=1]$. 
In order to compute/estimate the probability of $P_1$ being an independent set we observe the following. Let $P_1=\{p_1,\dots,p_k\}$. For each point $p_i\in P_1$, let $D(p_i,r)$ and $D(p_i,r/2)$ be the two disks of radius $r$ and $r/2$, respectively, which are centered at $p_i$. For each $i=2,\dots, k$, let $A_i$ be the event that $D(p_i, r/2)$ is disjoint from all the disks $D(p_j,r/2)$, for $j=1,\dots, i-1$, i.e.,
$$D(p_i,r/2)\cap D(p_j,r/2)=\emptyset, \quad\text{ for all } j=1,\dots, i-1.$$

See Figure~\ref{fig:is-disjoint1}. Observe that if $P_1$ is an independent set, then all the disks $D(p_j,r/2)$, with $j\in\{1,\dots, k\}$ are pairwise disjoint. Therefore, 
$$\text{if } P_1 \text{ is an independent set, then }(A_1\land A_2\land\dots\land A_k).$$ 

Thus, we have

\begin{align}
\Pr[X_1=1]&\le\notag \Pr[A_1\land \dots\land A_k]\\
&= \Pr[A_1]\label{is-X1-Ak}\cdot \Pr[A_2| A_1]\cdot \Pr[A_3| A_2\land A_1] \cdots \Pr [A_k| A_{k-1}\land\dots\land A_{1}].
\end{align}

\begin{figure}[H]
  \centering
\includegraphics[width=.5\columnwidth]{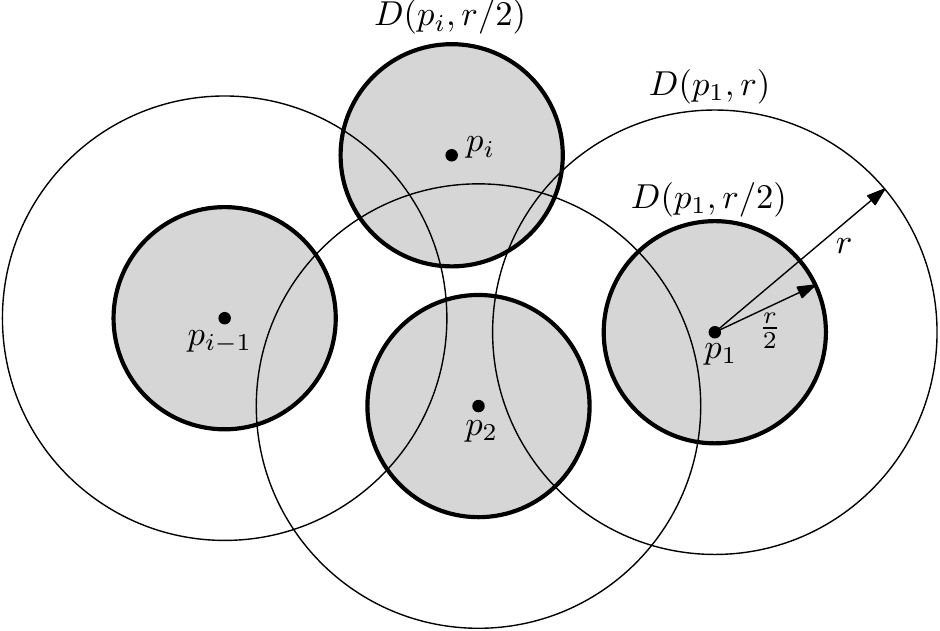}
\caption{$p_i$ should be outside $D(p_j,r)$ for all $j=1,\dots, i-1$.}
\label{fig:is-disjoint1}
\end{figure}

Note that $\Pr[A_1]=1$, and we are going to compute $\Pr[A_i|A_{i-1}\land\dots\land A_1]$ for $i=2, \dots, k$.

In order for $D(p_i, r/2)$ to be disjoint from $\bigcup_{j=1}^{i-1}D(p_j,r/2)$, $p_i$ should lie outside $\bigcup_{j=1}^{i-1}D(p_j,r)$. 
\begin{align}
\Pr[A_i|A_{i-1}\land\dots\land A_1]&=\notag\Pr\left[p_i\notin \bigcup_{j=1}^{i-1}D(p_j,r) \given[\Big] A_{i-1}\land\dots\land A_1\right]\\
 &= \notag1-  \Pr\left[p_i\in \bigcup_{j=1}^{i-1}D(p_j,r) \given[\Big] A_{i-1}\land\dots\land A_1\right]\\
&\le\notag  1-\Pr\left[p_i\in \bigcup_{j=1}^{i-1}D(p_j,r/2)\given[\Big] A_{i-1}\land\dots\land A_1\right]\\
&=\notag 1-(i-1) \pi (r/2)^2\\
&\le\label{is-Ai} e^{-\frac{(i-1) \pi r^2}{4}}.
\end{align}
By Inequalities~\eqref{is-X1-Ak} and~\eqref{is-Ai} we have
\begin{align}
\Ex[X_1]&\le \notag \Pr[A_1\land\dots\land A_k]\\
&\le\notag 1\cdot e^{-\frac{\pi r^2}{4}}\cdot e^{-\frac{2 \pi r^2}{4}}\dots e^{-\frac{(k-1)\pi r^2}{4}}\\
&=\notag e^{-\frac{\pi r^2}{4}(1+2+\dots+(k-1))}\\
& = \notag e^{-\frac{\pi r^2}{8} k(k-1)}\\
&\le\notag e^{-\frac{\pi r^2k^2}{16}},
\end{align}
where the last inequality is valid because $k-1>\frac{k}{2}$.
This and Equality~\eqref{is-E-X-X1} imply that 
$$\Ex[X]\le {n \choose k}\cdot e^{-\frac{\pi r^2k^2}{16}}\le
\left(\frac{en}{k}\right)^k\cdot e^{-\frac{\pi r^2k^2}{16}}.$$

Note that $r>\sqrt{\frac{c\ln(\frac{en}{k})}{k}}$. By Markov's Inequality we have:
\begin{align*}
\Pr[X\geq 1]
&\leq E[X] \\
& \leq \left(\frac{en}{k}\right)^k e^{-\frac{\pi r^2k^2}{16}}\\
& \leq \left(\frac{en}{k}\right)^k e^{- \frac{\pi c k}{16} \ln(\frac{en}{k})}\\
&= \left(\frac{en}{k}\right)^k \left(\frac{en}{k}\right)^{- \frac{\pi c k}{16}}\\
&= \left(\frac{en}{k}\right)^{k(1- \frac{\pi c}{16})},
\end{align*}
which tends to zero if $c>\frac{16}{\pi}$ and $n\to \infty$. This implies that with high probability $X=0$, and hence $\dg$ has no independent set of size $k$.
\end{proof}

\begin{figure}[htb]
  \centering
\includegraphics[width=.25\columnwidth]{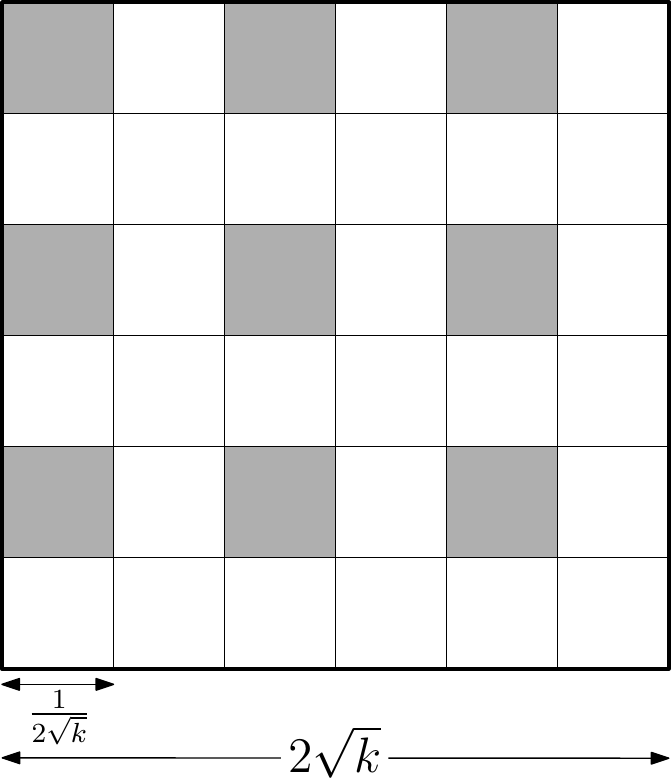}
\caption{The shaded squares belong to $S$.}
\label{fig:is-partition1}
\end{figure}

\begin{theorem}
\label{is-thr2}
Let $k\leq \frac{n}{4\ln n}$ be an integer. If $r \leq \frac{1}{2\sqrt{k}}$, then $\dg$ has an independent set of size $k$ with high probability.
\end{theorem}

\begin{proof}
Partition the unit square into sub-squares each of side length $\frac{1}{2\sqrt{k}}$. Note that this partition contains $4k$ sub-squares. Let $S$ be the set of $k$ sub-squares selected from every second column and every second row. See Figure~\ref{fig:is-partition1}. The probability that no point lies in a specific sub-square is $(1 - \frac{1}{4k})^n$.
Thus, the probability that there exists an empty sub-square
in $S$ is at most
$$
k \left(1 - \frac{1}{4k} \right)^n
\leq k e^{\frac{-n}{4k}} \leq \frac{1}{4\ln n},
$$
where the second inequality is valid because $k\leq \frac{n}{4\ln n}$. 
Therefore, with probability at least $1-\frac{1}{4\ln n}$ all sub-squares in $S$ contain points. On the other hand, since $r \leq \frac{1}{2\sqrt{k}}$, in $\dg$ there is no edge between two points in different sub-squares of $S$. Thus, by choosing one point from each sub-square in $S$ we obtain an independent set of size $k$. 
\end{proof}

\begin{corollary}
\label{is-cor1}
Let $k > \frac{n}{4 \ln n}$ be an integer. If $r > \sqrt{\frac{c\ln \ln n}{k}}$, for some constant $c>0$, then $\dg$ has no independent set of size $k$ with high probability.
\end{corollary}

\begin{proof}
Since $k > \frac{n}{4\ln n}$, we have $\ln(\frac{en}{k})\le 1+\ln 4 + \ln \ln n$. Thus, the statement follows by Theorem~\ref{is-thr1}. 
\end{proof}

\section{Concluding Remarks}
We presented thresholds for a random geometric graph, $\dg$, to have a connected subgraph or a clique of constant size, to be plane, and to be planar. We also investigated the existence of free edges and independent sets in $\dg$. Following are natural extensions of the problems discussed in this paper that are worth pursuing:
\begin{enumerate}
\item Since any monotone property of $\dg$ has a sharp threshold function \cite{friedgut1996every, goel2004sharp}, it would be interesting to provide such functions for the above properties. 
	
 \item Having a connected subgraph or a clique of size $k$ are monotone properties. For any integer constant $k\ge2$, we proved that $n^{\frac{-k}{2k-2}}$ is a threshold function for having a connected subgraph or a clique of size $k$. A natural problem is to extend these results for the case where $k$ is not necessarily a constant.
  \item  Having an independent set of size $k$ is a monotone property in $\dg$. We provided lower and upper bounds on $r$ for the existence of an independent set of size $k$. A natural problem is to improve any of the provided bounds.

  \item The existence of a free edge in $\dg$ is not a monotone property. We proved that if $r=o\left(\frac{1}{n}\right)$ or $r=\Omega\left(\frac{(\ln n)^{3/4}}{n^{1/4}}\right)$, then w.h.p. $\dg$ has no free edge. Moreover, if $r=\omega\left(\frac{1}{n}\right)$ and $r=o\left(\frac{1}{n^{1/2+\epsilon}}\right)$ for any constant $\epsilon>0$, then w.h.p. $\dg$ has a free edge. The threshold behavior for the existence of a free edge in $\dg$ when $r$ belongs to the interval $\left[\frac{1}{n^{1/2+\epsilon}},\frac{(\ln n)^{3/4}}{n^{1/4}}\right]$ remains as an open problem.
\item We provided examples of unit disk graphs that do not have any free edge. An interesting question is to determine if every planar unit disk graph has at least one free edge.
 
\end{enumerate}

\bibliographystyle{abbrv}
\bibliography{biblio}

\end{document}